\documentclass{llncs}
 
 \usepackage{algorithmicx}
\usepackage{amsmath}
 \usepackage{amssymb}
\usepackage{algorithm}
\usepackage{algpseudocode}
\usepackage{microtype}
\usepackage{graphicx}

\newtheorem{observation}{Observation} 
\newtheorem{clm}{Claim}

\newcommand{\evc}[0]{\operatorname{evc}}
\newcommand{\mvc}[0]{\operatorname{mvc}}

\title{A new lower bound for eternal vertex cover number}
\titlerunning{A new lower bound for eternal vertex cover number} 

\author{Jasine Babu \and Veena Prabhakaran}
\institute{Indian Institute of Technology Palakkad, India \email{jasine@iitpkd.ac.in,111704003@smail.iitpkd.ac.in}}

\authorrunning{J.\,Babu, V.\.Prabhakaran} 

\begin{document}
\maketitle
\begin{abstract}
We obtain a new lower bound for the eternal vertex cover number of an arbitrary graph $G$, in terms of the cardinality of a vertex cover of minimum size in $G$ containing all its cut vertices.  The consequences of the lower bound includes a quadratic time algorithm for computing the eternal vertex cover number of chordal graphs.
\end{abstract}
\section{Introduction}
Eternal vertex cover number of a graph is the minimum number of guards required to successfully keep defending attacks on a graph,
in a certain multi-round attack-defense game \cite{Klostermeyer2009}. The rules to play the game with $k$ guards on a graph $G$ are the following. Initially, the defender
places the $k$ guards on a subset of vertices of $G$. The positions of the guards defines an \textit{initial configuration}.
In each round of the game, the attacker chooses an edge $e$ of $G$ to attack. In response to the attack, the defender is 
free to move each of guard from its current position to an adjacent vertex or retain it in its current position. All guards are assumed to move in parallel, at the same time. 
The constraint to be satisfied is that at least one guard should move from an endpoint of $e$ to the other. 
If the defender is able to successfully move the guards satisfying this constraint, we say that the attack in the current round is successfully
defended. The resultant positions of the guards define the \textit{configuration} from where the next round of the attack-defense game continues. 
If the defender can keep on successfully defending any sequence of attacks, we say that the defender has a defense strategy on this graph, with $k$ guards. 
Eternal vertex cover number of a graph $G$, denoted by $\evc(G)$ is the minimum integer $k$ such that the defender has a defense strategy on $G$, with $k$ guards. 
When this game is played with $k$-guards, each configuration encountered in the game is equivalent to some function $f$ from $V$ to $\{0, 1, 2, \ldots, k\}$ 
such that $\sum_{v \in V} f(v)=k$ (where, for each $v \in V$, $f(v)$ will be the number of guards on $v$). 
A set of such configurations $\mathcal{C}$, such that the defender can start with any configuration in $\mathcal{C}$ as the initial configuration and keep moving between 
configurations in $\mathcal{C}$ for defending the attacks, is called an \textit{eternal vertex cover class} of $G$ and
each configuration in $\mathcal{C}$ is an \textit{eternal vertex cover configuration}. 
If $\mathcal{C}$ is an eternal vertex cover class of $G$ such that the number of guards in the configurations in $\mathcal{C}$ 
is equal to $\evc(G)$, then $C$ is a \textit{minimum eternal vertex cover class} of $G$.  
There are two popular versions of the game: the former in which in any configuration, at most one guard is allowed on a vertex and the latter in which this restriction is not there.
Since the main structural result in this paper is a lower bound for eternal vertex cover number, 
we will be assuming the version of the game in which there is no restriction on the number of guards allowed on a vertex. It can be easily verified
that our proofs work the same way in the other model of the game as well.

From the description of the game, it is clear that, in any configuration, if at least one of the endpoints of an edge is not occupied, the defender
will not be able to successfully defend an attack on that edge. Therefore, $\mvc(G) \le \evc(G)$, where $\mvc(G)$ denotes the cardinality of a minimum 
vertex cover of $G$. This is the only general lower bound known for the parameter, so far in literature. 
In this work, we prove that the size of a minimum sized vertex cover of $G$ that contains all cut vertices of $G$ is also a lower bound for $\evc(G)$. 
This improved lower bound has many algorithmic consequences, including a quadratic time algorithm for computing the eternal vertex cover number of chordal graphs and a 
PTAS for computing the eternal vertex cover number of internally triangulated planar graphs. These results generalize the results presented in \cite{BCFPRW:Arxiv}.
\section{A new lower bound}
\begin{definition}[$x$-components and $x$-extensions]
 Let $x$ be a cut vertex in a graph $G$ and $H$ be a component of $G \setminus x$. Let 
 $G'$ be the induced subgraph of $G$ on the vertex set $V(H) \cup \{x\}$. Then $G'$ is called an $x$-component of $G$ and $G$ is called an $x$-extension of $G'$. 
\end{definition}

Let $G'$ be a graph and $G$ be an $x$-extension of $G'$ for some $x \in V(G')$. 
It is easy to see that in every eternal vertex cover configuration of $G$ at least $\mvc(G')$ guards are
present on $V(G')$. However,
it is interesting to note that it is possible to have less than $\evc(G')$ guards present on $V(G')$
in some eternal vertex 
cover configurations of $G$. Hence,  though 
a lower bound for $\evc(G)$ can be obtained in terms of the minimum vertex cover numbers of the 
$x$-components of $G$, it may not be possible to obtain a non-trivial lower bound for $\evc(G)$ in 
terms of the eternal vertex cover numbers of the $x$-components of $G$.
Here, we introduce a new parameter and show that it is a lower bound for $\evc(G)$.
\begin{definition}
 Let $G$ be a graph and $X \subseteq V(G)$. The smallest integer $k$, such that $G$ has a vertex cover $S$ of cardinality $k$ with 
 $X \subseteq S$, is denoted by $\mvc_{X}(G)$.  
\end{definition}
To simplify the expressions that appear later, we introduce the following notations.
For any vertex $v \in V(G)$, $\mvc_{\{v\}}(G)$ will be denoted by $\mvc_v(G)$
and for any graph $G$ and any set $X$, the notation $X(G)$ will be used to denote the 
set $X \cap V(G)$.
\begin{definition}
 Let $X$ be the set of cut vertices of a graph $G$ and let $x \in X$. The set of $x$-components of $G$ will be denoted as $\mathcal{C}_x(G)$. 
 If $B$ is any block of $G$, then the set of $B$-components of $G$ is defined as\\
 $\mathcal{C}_B(G)= \left \{ G_i: G_i \in \mathcal{C}_x(G) \text{ for some $x \in X(B)$} \text{ and $G_i$ edge disjoint with $B$}\right\}$.
\end{definition}

%

\begin{definition}[EVC-Cut-Property]\label{def:evc-cut-property}
Let $G'$ be a graph and let $X'$ be the set of cut vertices of $G'$.
The graph $G'$ is said to have the EVC-cut-property if for every graph $G$ that is an $x$-extension of $G'$ for some $x \in V(G')$, it is true that 
in each eternal vertex cover configuration of $G$, at least 
$\mvc_{_{X' \cup \{x\}}}(G')$ guards are present on the vertices of $G'$, out of which at least $\mvc_{_{X' \cup \{x\}}}(G')-1$ guards 
are present on $V(G') \setminus \{x\}$. 
\end{definition}
\begin{note}
For a graph $G'$  to satisfy the EVC-Cut-Property, it is not necessary that the vertex $x$ is occupied by a guard in every eternal 
vertex cover configuration of an $x$-extension $G$ of $G'$. 
All the $\mvc_{_{X' \cup \{x\}}}(G')$ (or more) guards could be on vertices other than $x$. 
\end{note}
\begin{note}
 Definition~\ref{def:evc-cut-property} gives some lower bounds on the number of guards and not on the number of vertices with guards. 
 Note that, if more than one guard is allowed on a vertex, then these two numbers could be different.  
\end{note}

The following two lemmas are easy to obtain, using a straightforward counting argument. 

\begin{lemma}\label{lem:cut-sum}
 Let $G$ be a graph and $X$ be the set of cut vertices of $G$. For any $x \in X$, \ 
 \[ \mvc_{_{X\cup \{x\}}}(G)= \mvc_{_X}(G)=1+\sum_{G_i \in \mathcal{C}_x(G)}{\left[\mvc_{_{X(G_i)}}\left(G_i\right) -1\right]}.\]
\end{lemma}

\begin{lemma}\label{lem:block-sum}
 Let $G$ be a graph and $X$ be the set of cut vertices of $G$. 
 If $B$ is a block of $G$ and $v$ is any vertex of $B$ such that $v \notin X(B)$, then\\ 
 \[ \mvc_{_{X\cup \{v\}}}(G)=\mvc_{_{X(B) \cup \{v\}}}(B)+\sum_{G_i \in \mathcal{C}_B(G)}{\left[\mvc_{_{X(G_i)}}(G_i) -1\right]}.\]
\end{lemma}
\begin{lemma}\label{lem:evc-cut-property}
 Every graph satisfies EVC-cut-property. 
\end{lemma}
\begin{proof}
The proof is by induction on the number of blocks of the graph. 
 First consider a graph $G'$ with a single block. Let $x$ be any vertex of $G'$ and $G$ be an $x$-extension of $G'$. 
 Let $C$ be an eternal vertex cover configuration of $G$ and let $S$ be the set of vertices of $G$ on which guards are present in $C$.
 Since $C$ is an eternal vertex cover configuration of $G$, $S$ must be a vertex cover of $G$ and $S \cap V(G')$ must be a vertex cover of $G'$. 
 Therefore, $|S \cap V(G')| \ge \mvc(G')$.
  If $|S \cap V(G')| \ge \mvc_{x}(G')$, then there are at least $\mvc_{x}(G')$ guards on $V(G')$ and  at least $\mvc_{x}(G')-1$ guards on 
 $V(G') \setminus \{x\}$, as we need to prove.  Also, it is easy to see that $\mvc_{x}(G') \le  \mvc(G')+1$. 
 Therefore, we are left with the case when $\mvc(G')=|S \cap V(G')| < \mvc_{x}(G')=\mvc(G')+1$.  This implies that $x \notin S$.
 Thus, in the remaining case to be handled, the number of vertices on which guards are present is exactly $\mvc(G')$ and there is no guard on $x$. 
 
 From this point, let us focus on the number of guards on $V(G')$ and not just the number of vertices that are occupied. If there are
 more than $\mvc(G')$ guards in $V(G')$, then the conditions we need to prove are satisfied for the configuration $C$.
 In the remaining case, we have exactly $|S \cap V(G')|=\mvc(G')$ guards in $V(G')$, with $x \notin S$. In this case, we will derive a contradiction.
 
 Consider an attack on an edge $xv$ incident at $x$, where $v \in V(G')$. 
Let $\tilde{C}$ be the new configuration, after defending this attack and $\tilde{S}$ be the set of vertices on which guards are present in $\tilde{C}$. 
In the transition from $C$ to $\tilde{C}$, a guard must have moved from $v$ to $x$. Also, $x$ being a cut vertex, 
no guard can move from $V(G) \setminus V(G')$ to $V(G') \setminus \{x\}$. Therefore, $|\tilde{S} \cap V(G')|= |S \cap V(G')| = \mvc(G')$. 
But, this is a contradiction because $\tilde{S} \cap V(G')$ is a minimum vertex cover of $G'$ containing $x$, but we have $\mvc(G') < \mvc_{x}(G')$. 

Thus, the lemma holds for all graphs with only one block.  
Now, as induction hypothesis, assume that the lemma holds for any graph $G'$ with at most $k$ blocks. We need to show that the lemma 
holds for any graph with $k+1$ blocks. 
 
 Let $G'$ be an arbitrary graph with $k+1$ blocks and let $x$ be an arbitrary vertex of $G'$. Let $X'$ be the set of cut vertices of 
 $G'$ and let $G$ be an arbitrary $x$-extension of $G'$. Let $C$ be an arbitrary eternal vertex cover configuration of $G$ and let $S$
 be the set of vertices on which guards are present in $C$. 
 Let $l=\mvc_{_{X'\cup \{x\}}}(G')$.
 We need to show that there are at least $l$ guards on $V(G')$ in $C$ and at least $l-1$ guards on  $V(G') \setminus \{x\}$. Let $t$ be
 the number of guards on $V(G')$ in $C$.
 We split our proof into two cases based on whether $x$ is a cut vertex in $G'$ or not. \\\\
\textbf{Case 1. } $x$ is a cut vertex of $G'$:

In this case, by our induction hypothesis, for each $x$-component $G_i$ of $G'$, at least $\mvc_{_{X'(G_i)}}\left(G_i\right)$ guards are 
on $V(G_i)$ in the configuration $C$.
There are two possible sub-cases.
\begin{enumerate}
 \item[(a)] If $x$ is not occupied by a guard in $C$, then by induction hypothesis,\\ $t \ge \sum_{G_i \in \mathcal{C}_x(G')}{\mvc_{_{X'(G_i)}}\left(G_i\right)}$. 
Since $\mathcal{C}_x(G')$ is non-empty, by Lemma~\ref{lem:cut-sum}, it follows that $t \ge \mvc_{X'\cup \{x\}}(G')=l$. 
Since $x$ is not occupied, the number of guards on  $V(G') \setminus \{x\}$ is $t$ itelf, where $t \ge l$, as shown.
\item[(b)] If $x$ is occupied by a guard in $C$, still, in order to satisfy the induction hypothesis for all $x$-components of $G'$, 
the number of guards on $V(G') \setminus \{x\}$ must be at least $\sum_{G_i \in \mathcal{C}_x(G')}{\left(\mvc_{_{X'(G_i)}}\left(G_i\right) -1\right)}$. 
Therefore, by Lemma~\ref{lem:cut-sum}, it follows that the number of guards on $V(G') \setminus \{x\}$ is at least $l-1$ and $t \ge l$.
\end{enumerate}
\textbf{Case 2. } $x$ is not a cut vertex of $G'$:

Let $B$ be the block of $G'$ that contains $x$. By Lemma~\ref{lem:block-sum}, we have:
\begin{equation}\label{eq1}
 l= \mvc_{_{X'(B) \cup \{x\}}}(B)+\sum_{G_i \in \mathcal{C}_B(G')}{\left(\mvc_{_{X'(G_i)}}(G_i) -1\right)}
\end{equation}
Before proceeding with the proof, we establish the following claim.
  \begin{clm}\label{clm:clm1}
 Suppose $C'$ is an eternal vertex cover configuration of $G$. Then the number of guards on $V(G') \setminus \{x\}$ in configuration $C'$ is at least $l-1$. 
  \end{clm}
\begin{proof} 
    ~ To count the number of guards on $V(G') \setminus \{x\}$, we count the total number of guards on the $B$-components of $G'$ and
    the number of guards on the remaining vertices separately and add them up. 
    \begin{itemize}
    \item First, we will count the total number of guards on the $B$-components of $G'$. For each $B$-component $G_i$ of $G'$, let 
    $k_i=\mvc_{_{X'(G_i)}}(G_i)$. For each cut vertex $v \in  X'(B)$, 
    let $\mathcal{C}_v$ denote the family of $B$-components of $G'$ that intersect at the cut vertex $v$ and let $n_v$ denote 
 $|\mathcal{C}_v|$. 
  Consider a $B$-component $G_i$ of $G'$. 
  By our induction hypothesis, the number of guards on $V(G_i)$ is at least $k_i$ in $C'$. Moreover, since $G_i$ is connected to $B$
  by a single cut vertex, from the induction hypothesis it follows 
  that the number of guards on $V(G_i) \setminus B$ is at least $k_i-1$. Note that, for each cut vertex $v \in X'(B)$, the total number
  of guards on $\bigcup_{G_i \in C_v} V(G_i)$ must be at least 
  $1 + \sum_{i : G_i \in C_v} (k_i -1)$, to satisfy the above requirement. By summing this over all the cut vertices in $X'(B)$, the total
  number of guards on $\bigcup_{G_i \in \mathcal{C}_B(G')} V(G_i)$ must be at 
  least $|X'(B)|+ \sum_{G_i \in \mathcal{C}_B(G')}{\left(\mvc_{_{X'(G_i)}}(G_i) -1\right)}$. 
   \item  Now, we will count the number of guards on the remaining vertices. To cover the edges inside the block $B$ that are not 
   incident at any vertex in $X'(B)$, at least $\mvc(B \setminus X'(B))$ 
   vertices of $B \setminus X'(B)$ are to be occupied in $C'$. If $x$ is occupied in $C'$, then at least $\mvc_x(B \setminus X'(B))$ 
   vertices of $B \setminus X'(B)$ are occupied in $C'$.
   Hence, irrespective of whether $x$ is occupied in $C'$ or not, the number of guards on $(V(B) \setminus X'(B)) \setminus \{x\}$ is at least $\mvc_x(B \setminus X'(B))-1$.    
   \end{itemize}
 Therefore, the total number of guards on $V(G') \setminus \{x\}$ is at least 
 $\mvc_x(B \setminus X'(B))-1 + |X'(B)|+ \sum_{G_i \in \mathcal{C}_B(G')}{\left(\mvc_{_{X'(G_i)}}(G_i) -1\right)}$. Since 
 $\mvc_x(B \setminus X'(B))+ |X'(B)|=\mvc_{_{X'(B)\cup \{x\}}}(B)$, we can conclude 
 that the number of guards on $V(G') \setminus \{x\}$ is equal to  
 $\mvc_{_{X'(B)\cup \{x\}}}(B)-1+ \sum_{G_i \in \mathcal{C}_B(G')}{\left(\mvc_{_{X'(G_i)}}(G_i) -1\right)}$. Comparing this expression 
 with Equation~(\ref{eq1}), we can see that the number of guards on $V(G') \setminus \{x\}$ is at least $l-1$. 
  \qed
  \end{proof}
 Now, we continue with the proof of Lemma~\ref{lem:evc-cut-property}. There are two possible sub-cases.
 \begin{enumerate}
  \item[(a)]  If $x$ is occupied by a guard in $C$, then by Claim~\ref{clm:clm1}, it follows that the number of guards on
  $V(G') \setminus \{x\}$ is at least $l-1$ 
  and the number of guards on $V(G')$ is at least $l$, as we require.
  \item[(b)] If $x$ is not occupied in $C$, then by Claim~\ref{clm:clm1}, $t \ge l-1$. 
  If $t \ge l$, we are done. If $t=l-1$, then we will derive a contradiction. 
  Consider an attack on an edge $xu$ such that $u \in V(B)$. 
  While defending this attack, a guard must move from $u$ to $x$. 
  Let $\tilde{C}$ be the new configuration in $G$ and let $\tilde{S}$ be the set of vertices on which guards are present in $\tilde{C}$. 
  Note that no guards from $V(G) \setminus V(G')$ can move to any vertex of $V(G') \setminus \{x\}$ in this transition from 
  $C$ to $\tilde{C}$, because $x$ is a cut vertex in $G$. 
  Therefore, in $\tilde{C}$, the total number of guards on $V(G') \setminus \{x\}$ is less than $l-1$, contradicting 
  Claim~\ref{clm:clm1}. Therefore, $t=l$ and the lemma holds for $G'$.
 \end{enumerate}
  Thus, by induction, the lemma holds for every graph. 
  \qed
\end{proof}
\begin{remark}
The above lemma holds for both the models of the eternal vertex cover problem; the first model in which the number of 
guards permitted on a vertex in any configuration is limited to one and the 
second model, where this restriction is not there. However,
it is possible that, in the second model, the number of vertices on which guards are present could be smaller than 
$\mvc_{_{X\cup \{x\}}}(G)$ in some valid configurations. An example illustrating this subtlety is shown in Figure~\ref{fig:figure1}.
In order to address this subtlety, the proof of Lemma~\ref{lem:evc-cut-property} employs a careful interplay between the two quantities a) the number of guards 
in a configuration and b) the number of vertices on which guards are present in a configuration. 
\end{remark}
\begin{figure}[h] 
\centering
\includegraphics[scale=0.7]{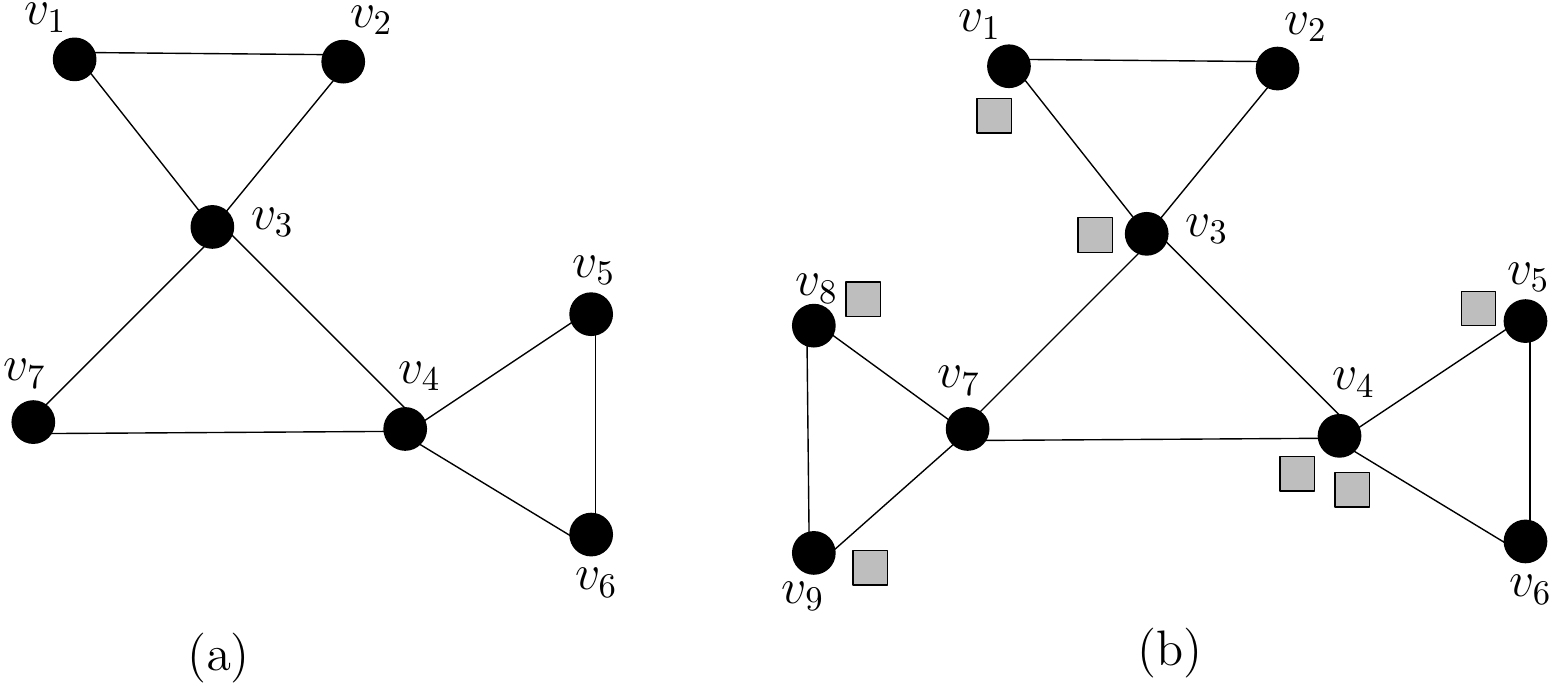}
\caption{ Any vertex cover of the graph in (a) that contains vertex $v_7$ and both the cut vertices must be of size at least $5$. 
The graph in (b) is a $v_7$-extension of the graph in (a). Positions of guards in an eternal vertex cover configuration of 
the graph in (b) are indicated using gray squares. 
This is a valid configuration. Note that, only four vertices of the graph in (a) are occupied in the configuration shown in (b).}
 \label{fig:figure1}
 \end{figure}
 
\begin{theorem}\label{thm:lower-bound}
 For any connected graph $G$, $\evc(G) \ge \mvc_{_X}(G)$, where $X$ is the set of cut vertices of $G$. 
\end{theorem}
\begin{proof}
Let $C$ be an eternal vertex cover configuration of $G$ and $S$ be the set of all vertices of $G$ containing guards in $C$. 
Suppose $\evc(G)<\mvc_{_X}(G)$. Then, there exists a vertex $x \in X$ such that $x \notin S$. Since every graph satisfies EVC-cut-property
by Lemma~\ref{lem:evc-cut-property}, for each $x$-component $G_i$ of $G$, 
exactly $\mvc_{_{X(G_i)}}(G_i)$ guards are present on $V(G_i) \setminus \{x\}$. 
Therefore, the total number of guards is at least $\sum_{G_i \in \mathcal{C}_x(G)}{\mvc_{_{X(G_i)}}\left(G_i\right)}$. Since there are
at least two $x$-components, by comparing this expression with the RHS of the equation in Lemma~\ref{lem:cut-sum}, we can see that
the total number of guards is more than $\mvc_X(G)$. This contradicts our initial assumption.
\qed
\end{proof}
\begin{observation}\label{obs:cutvertex-occupy}
Let $G$ be a connected graph and let $X$ be the set of cut vertices of $G$. 
If $\evc(G)=\mvc_{_X}(G)$, then in every minimum eternal vertex cover configuration of $G$, there are guards on each vertex of $X$.  
\end{observation}
\begin{proof}
  For contradiction, assume that there exists a minimum eternal vertex cover configuration $C$ of $G$
with a cut vertex $x$ unoccupied. Rest of the proof is exactly the same as in the proof of Theorem~\ref{thm:lower-bound}.  
\qed
\end{proof}
 \section{Algorithmic implications}
In this section, we first prove some general implications of Theorem~\ref{thm:lower-bound}, which are used for deriving algorithmic results for some well-known graph classes.
\begin{definition}[Graph class\footnote{Note that the definition of this graph class is more general than the one in \cite{BCFPRW:Arxiv}.} $\mathcal{F}$]
$\mathcal{F}$ is defined as the family of all connected graphs $G$ satisfying the following property: if $X$ is the 
set of cut vertices of $G$ and $S$ is any vertex cover of $G$ with $X \subseteq S$ and $|S|=\mvc_{_X}(G)$, 
then the induced subgraph $G[S]$ is connected.
\end{definition}
For any graph $G$ and $S \subseteq V(G)$, let $\evc_{_S}(G)$ denote the minimum number $k$ such that $G$ has an eternal vertex cover class $\mathcal{C}$ with $k$ guards in which all vertices of $S$ are 
occupied in every configuration of $\mathcal{C}$. 
For an example, let $G$ be a path on three vertices $u$,$v$ and $w$, in which $v$ is the degree-two vertex. It can be easily seen that $\evc(G)=2$. Since \{\{u,v\}, \{v,w\}\} is an eternal vertex cover class of $G$ in which each configuration has $v$ occupied, $\evc_{\{v\}}(G)=2$.
Since $G$ has no eternal vertex cover class in which each configuration contains $u$ and has exactly two vertices, it follows that $\evc_{\{u\}}(G)=3$.
By Observation~\ref{obs:cutvertex-occupy}, we have the following 
generalization of Corollary~2 of \cite{BCFPRW:CALDAM}.
\begin{theorem}\label{thm:evc-subset}
 Let $G$ be a graph in $\mathcal{F}$ with at least two vertices and $X$ be the set of cut vertices of $G$. 
  If for every vertex $v \in V(G) \setminus X$, $\mvc_{_{X \cup \{v\}}}(G)=\mvc_{_X}(G)$, then $\evc(G)=\evc_{_X}(G)=\mvc_{_{X}}(G)$.
 Otherwise,  $\evc(G) =\evc_{_X}(G)= \mvc_{_X}(G)+1$.
\end{theorem}
\begin{proof}
By Theorem~\ref{thm:lower-bound}, we have $\mvc_{_X}(G) \le \evc(G)$ and we have $\evc(G) \le \evc_{_X}(G)$. 
\begin{itemize}
 \item If for every vertex $v \in V(G) \setminus X$, $\mvc_{_{X \cup \{v\}}}(G)=\mvc_{_X}(G)$, then by Lemma~2 of \cite{BCFPRW:CALDAM},
 $\evc_{_X}(G)=\mvc_{_X}(G)$ and hence, $\evc(G)=\evc_{_X}(G)=\mvc_{_{X}}(G)$. 
 \item If for some vertex $v \in V(G) \setminus X$, $\mvc_{_{X \cup \{v\}}}(G) \ne \mvc_{_X}(G)$, then by Theorem~1 of 
 \cite{BCFPRW:CALDAM}, $\evc_{_X}(G) \ne \mvc_{_X}(G)$. 
  Let $S$ be any minimum sized vertex cover of $S$ that contains all vertices of $X$. Since $S$ is a connected vertex cover of $G$,
 by a result by Klostermeyer and Mynhardt~\cite{Klostermeyer2009}, $\evc(G) \le \evc_{_X}(G) \le |S|+1=\mvc_{_X}(G)+1$. Thus, we have 
  $\mvc_{_X}(G)<\evc(G)=\evc_{_X}(G)=\mvc_{_{X}}(G)+1$. 
 \end{itemize}
 \qed
 \end{proof}
 The following corollary is a generalization of Remark~3 of \cite{BCFPRW:Arxiv}.
 \begin{corollary}\label{cor:compute-evc}
 Let $G$ be a graph in $\mathcal{F}$ with at least two vertices and $X$ be the set of cut vertices of $G$. 
 Then, $\evc(G) = \min \{ k: \forall v \in V(G),$ $G$ has a vertex cover $S_v$ of size $k$
 such that $X \cup \{v\} \subseteq S_v\}$.
\end{corollary}
 The following result is a generalization of Corollary~3 of \cite{BCFPRW:Arxiv}.
\begin{observation}\label{obs:evc-bic-np}
 Given a graph $G \in \mathcal{F}$ and an integer $k$, deciding whether $\evc(G) \le k$ is in NP.
\end{observation}
\begin{proof}
Consider any $G \in \mathcal{F}$ with at least two vertices and let $X$ be the set of cut vertices of $G$.
 By Corollary~\ref{cor:compute-evc}, $\evc(G) = \min \{ k: \forall v \in V(G),\text{ $G$ has a vertex cover}\\ \text{$S_v$ of size $k$
 such that } X \cup \{v\} \subseteq S_v\}$.  
 To check if $\evc(G) \le k$, the polynomial time verifiable certificate consists of at most $|V|$ vertex covers of size at most $k$ such that for each vertex 
$v\in V$, there exists a vertex cover in the certificate containing all vertices of $ X \cup \{v\}$. 
\qed
\end{proof}
\subsection{Graphs with locally connected blocks}
A graph $G$ is \textit{locally connected} if for every vertex $v$ of $G$, its open neighborhood $N_G(v)$ induces a connected 
 subgraph in $G$ \cite{chartrand1974}. Biconnected chordal graphs and biconnected internally triangulated graphs are some well-known examples of locally
 connected graphs. If every block of a graph $G$ is locally connected, then every vertex cover of $G$ that contains all its cut vertices is connected.
 Hence, $G \in \mathcal{F}$ and by Theorem~\ref{thm:evc-subset}, we have:
\begin{corollary}\label{cor:loc-conn-bound}
 Let $G$ be a connected graph with at least two vertices, such that each block of $G$ is locally connected and let $X$ be the set of cut vertices of $G$. 
 Then, $\mvc_{_X}(G) \le \evc(G) \le \mvc_{_X}(G)+1$. 
 Further, $\evc(G)=\mvc_{_{X}}(G)$ if and only if for every vertex $v \in V(G) \setminus X$, $\mvc_{_{X \cup \{v\}}}(G)=\mvc_{_X}(G)$. 
 In particular, these conclusions hold for chordal graphs and internally triangulated planar graphs that are connected and have at least two vertices.
\end{corollary}
\subsection{Hereditary graph classes}
The following theorem is obtained by generalizing Theorem~3 of \cite{BCFPRW:CALDAM}, by applying Theorem~\ref{thm:evc-subset}.
\begin{theorem}\label{thm:heriditary}
 Let $\mathsf{C}$ be a hereditary graph class such that each biconnected graph in $\mathsf{C}$ is locally connected. 
 If the vertex cover number of any graph in $\mathsf{C}$ can be computed in $O(f(n))$ time, then the eternal vertex cover number
 of any graph $G \in \mathsf{C}$ can be computed in  $O(n.f(n))$ time.
\end{theorem}
\begin{proof}
 Let $G$ be a graph in $\mathsf{C}$. 
 Since each block of $G$ is locally connected, by Corollary~\ref{cor:loc-conn-bound}, $\mvc_{_X}(G) \le \evc(G) \le \mvc_{_X}(G)+1$. 
 Further, by Corollary~\ref{cor:loc-conn-bound}, to check whether $\evc(G) = \mvc_{_X}(G)$, it is enough to decide if 
 for every vertex $v \in V \setminus X$, $\mvc_{_{X \cup \{v\}}}(G)=\mvc_{_X}(G)$. Since minimum vertex cover computation can be done 
 for graphs of $\mathsf{C}$ in $O(f(n))$ time,
 for a vertex $v$, checking whether $\mvc_{_{X \cup \{v\}}}(G)=\mvc_{_X}(G)$, takes
 only $O(f(n))$ time. Therefore, checking whether $\evc(G) = \mvc_{_X}(G)$ can be done in $O(n. f(n))$ time. 
 \qed
\end{proof}
\subsection{Chordal graphs}
The following theorem is a special case of Theorem~\ref{thm:heriditary}, using the fact that minimum vertex cover computation 
can be done for chordal graphs in $O(m+n)$ time \cite{Tarjan},
 where $m$ is the number of edges and $n$ is the number of vertices of the input graph. This result is a generalization of a 
 result for biconnected chordal graphs in \cite{BCFPRW:Arxiv}.
\begin{theorem}\label{thm:chordal-evc}
 Let $G$ be a chordal graph and $X$ be the set of cut vertices of $G$. 
 Then, $\mvc_{_X}(G) \le \evc(G) \le \mvc_{_X}(G)+1$ and the value of $\evc(G)$ can be determined in $O(n^2+mn)$ time, where 
 $m$ is the number of edges and $n$ is the number of vertices of the input graph.
 \end{theorem}
\subsection{Internally triangulated planar graphs}
The following lemma is a generalization of a result in \cite{BCFPRW:Arxiv} for biconnected internally triangulated planar graphs.
\begin{lemma}
 Given an internally triangulated planar graph $G$ and an integer $k$, deciding whether $\evc(G) \le k$ is NP-complete.
\end{lemma}
\begin{proof}
 Since each block of an internally triangulated planar graph $G$ is locally connected, every vertex cover $S$ of $G$ that contains
 all its cut vertices induces a connected subgraph. Therefore, by Observation~\ref{obs:evc-bic-np}, deciding whether $\evc(G) \le k$ is in NP. 
 Since this decision problem is known to be NP-hard for biconnected internally triangulated graphs \cite{BCFPRW:Arxiv}, the lemma follows.  
 \qed
\end{proof}
The existence of a polynomial time approximation scheme for computing the eternal vertex cover number of biconnected internally triangulated planar graphs, 
given in \cite{BCFPRW:Arxiv}, is generalized by the following result.
\begin{lemma}
  There exists a polynomial time approximation scheme for computing the eternal vertex cover number of internally triangulated planar graphs.
 \end{lemma}
\begin{proof}
Let $G$ be an internally triangulated planar graph. 
Let $X$ be the set of cut vertices of $G$. It is possible to compute $X$ in linear time, using a  well-known depth first search based method. 
By Corollary~\ref{cor:compute-evc}, $\evc(G) = \max \{ \mvc_{X \cup\{v\}}(G):  v \in V(G) \}$.
It is easy to see that for a vertex $v\in V(G) \setminus X$, $\mvc_{X \cup\{v\}}(G)=|X|+1+\mvc({G \setminus (X\cup\{v\})})$.
For $v \in V(G)$, $\mvc_{X \cup \{v\}}(G)=\mvc_X(G)=|X|+\mvc({G \setminus X})$.
Using the PTAS designed by Baker et al.~\cite{Baker1994} for computing the vertex cover number of planar graphs, given any
$\epsilon >0$, it is possible to approximate $\mvc({G \setminus (X\cup\{v\})})$
within a $1+\epsilon$ factor, in polynomial time. From this, a polynomial time approximation scheme for computing $\evc(G)$ follows.
\qed
\end{proof}

\end{document}